\documentclass[12pt]{amsart}

\usepackage[utf8]{inputenc}

\usepackage{graphics}
\usepackage{amsfonts}
\usepackage{amssymb}
\usepackage{framed}
\usepackage{a4}
\usepackage{latexsym}
\usepackage{dsfont}
\usepackage{color}
\usepackage{mathrsfs} 
\usepackage{hyperref}
\usepackage{mathtools}
\usepackage{tensor}
\usepackage{ulem}
\usepackage[all]{xy}
\allowdisplaybreaks

\DeclareMathAlphabet{\eusm}{U}{}{}{}  % Euler script math
\SetMathAlphabet\eusm{normal}{U}{eus}{m}{n}
\SetMathAlphabet\eusm{bold}{U}{eus}{b}{n}

\DeclareMathAlphabet{\eufrak}{U}{}{}{}  % Euler fraktur math
\SetMathAlphabet\eufrak{normal}{U}{euf}{m}{n}
\SetMathAlphabet\eufrak{bold}{U}{euf}{b}{n}

\newtheorem{theorem}{Theorem}[section]
\newtheorem{proposition}[theorem]{Proposition}
\newtheorem{lemma}[theorem]{Lemma}
\newtheorem{corollary}[theorem]{Corollary}

\theoremstyle{definition}
\newtheorem{definition}[theorem]{Definition}

\theoremstyle{remark}

\numberwithin{equation}{section}

\newcommand{\C}{\mathbb{C}}
\newcommand{\Z}{\mathbb{Z}}
\newcommand{\N}{\mathbb{N}}

\title{Error Basis and Quantum Channel}

\author{B. V. Rajarama Bhat}
\address{Indian Statistical Institute,8th Mile, Mysore Road, RVCE Post, Bangalore 560 059, India}
\email{bvrajaramabhat@gmail.com}

\author{Purbayan Chakraborty}
\address{Laboratoire de math\'ematiques de Besan\c{c}on, Universit\'e de Franche-Comt\'e, 16, route de Gray, F-25 030 Besan\c{c}on cedex, France}
\email{purbayan.chakraborty@univ-fcomte.fr}

\author{Uwe Franz}
\address{Laboratoire de math\'ematiques de Besan\c{c}on, Universit\'e de Franche-Comt\'e, 16, route de Gray, F-25 030 Besan\c{c}on cedex, France}
\email{uwe.franz@univ-fcomte.fr}
\urladdr{http://lmb.univ-fcomte.fr/uwe-franz}

\begin{document}

\begin{abstract}
The Weyl operators give a convenient basis of $M_n(\C)$ which is also orthonormal with respect to the Hilbert-Schmidt inner product. 
The properties of such a basis can be generalised to the notion of a nice error basis(NEB), as introduced by E. Knill \cite{Knill}.
We can use an NEB of $M_n(\C)$ to construct an NEB for $\mathrm{Lin}(M_n(\C))$, the space of linear maps on $M_n(\C)$.
Any linear map will then correspond to a $n^2\times n^2$ coefficient matrix in the basis decomposition with respect to such an NEB of $\mathrm{Lin}(M_n(\C))$.
Positivity, complete (co)positivity or other properties of a linear map can be characterised in terms of such a coefficient matrix.

\end{abstract}

\maketitle

\section{Introduction}
The Pauli matrices had the key role for modelling quantum error correction for one qubit system. 
E. Knill generalised the properties of Pauli matrices in higher dimension to introduce the notion of nice error basis in 1996 \cite{Knill}.
Since then it has played a fundamental important role in quantum error correction.
A nice error basis(NEB) gives a convenient orthonormal basis of $M_n(\C)$ constructed from a projective representation of a group known as an index group.
Knill in his paper showed that the existence of such a basis of $M_n(\C)$ is equivalent to the existence of an irreducible character of a suitable group which vanishes outside its center.
The discrete Weyl relation \cite[Equation (46)]{weyl} $UV=qVU$, where $q$ is an $n$-th root of unity, and its n-dimensional realisation
\[
U e_k = e_{k+1 \, {\rm mod} \, n}, \quad Ve_k = q^k e_k, \quad k=0,1,\ldots,n-1,
\]
with $\{e_0,\ldots,e_{n-1}\}$ an orthonormal basis of $\mathbb{C}^n$, were discussed by Hermann Weyl when he studied the equivalence between Heisenberg's matrix mechanics and Schr\"odinger's wave mechanics in 1927. 
Knill used these discrete Weyl type operators to construct an NEB from the group $\Z^{2k}_p$ where $p$ is prime.
Later Klappenecker and R\"otteler \cite{kr1} and then Parthasarathy \cite{partha} constructed a similar NEB from the index group $\Z_n\times \Z_n$ and Parthasarathy used them for quantum error correcting codes.

Recently, X. Huang, T. Zhang et al \cite{Huang et al} have used the Weyl operators as a basis of $M_n(\C)$ to describe a quantum state and give a necessary condition for separability of a bipartite state. In this paper we take advantage of an NEB of $M_n(\C)$ in particular, Weyl operators on $M_n(\C)$, to construct a convenient basis of $\mathrm{Lin}(M_n(\C))$, the set of all linear maps on $M_n(\C)$.
Such a basis of $\mathrm{Lin}(M_n(\C))$ will be used to decompose a linear map on $M_n(\C)$, then positive and completely positive maps will be characterised with respect to the $n^2\times n^2$ coefficient matrix found in such decompositions. 
Moreover, we will see that the product of linear maps in $\mathrm{Lin}(M_n(\C))$ translates into the convolution type product of the corresponding coefficient matrices in $M_{n^2}(\C)$, which will be used to give a characterisation of completely co-positive maps on $M_n(\C)$.
Such an exploitation of isomorphims between $\mathrm{Lin}(M_n(\C))$ and $M_{n^2}$ or $M_n(\C)\otimes M_n(\C)^*$ may lead to an alternate approach to study different cones of positive maps (e.g. $k$-positive map, $k$-superpositive maps) compared to the existing theory using the Choi-Jamiolkowski isomorphism (cf. \cite{ssz}). 

In section $2$, we present the results of E. Knill\cite{Knill}, A. Klappenecker and Martin R\"otteler\cite{kr1}\cite{kr2} regarding the nice error basis and its construction from a group of central type. 
As a very useful example of NEB for computation, we will discuss the Weyl operators.
In section $3$, we will take advantage of the NEB to construct a convenient basis of $\mathrm{Lin}(M_n(\C))$ with respect to which we will decompose any linear map $\alpha \in \mathrm{Lin}(M_n(\C))$ to obtain the corresponding  $n^2\times n^2$ coefficient matrix $D_{\alpha}$.
We will compute the coefficient matrix for few maps which are important in quantum information e.g. depolarising channel, transposition and conditional expectation onto diagonals.
Section $4$ will be on the relation between the coefficient matrix $D_{\alpha}$ and the Choi matrix $C_{\alpha}$ corresponding to a linear map $\alpha$ on $M_n(\C)$.
Finally, section $5$ will be devoted to the characterisation of positivity properties of linear maps in terms of their coefficient matrices.
More precisely, we will characterise the positive, completely positive (in general quantum channel), completely co-positive maps on $M_n(\C)$ and a special case, $1$-super positive or entanglement breaking map in $M_2(\C)$.

%%%%%%%%%%%%%%%%%%%%%%%%%%%%%%%%%%%%%%%%%%%%%%%%%%%%%%%%%%%%%%%%%%%%%
\section{Nice error basis and ONB of $M_n(\mathbb{C})$}
In this section we briefly present the results of E. Knill \cite{Knill}, Klappenecker and Rotteler \cite{kr1} on nice error basis.
We recall the definition of a nice error basis as introduced by E. Knill \cite{Knill}.
\begin{definition}
    Let $G$ be a group of order $n^2$ for some natural number $n$.
    A \textit{nice error basis}(NEB) on $\mathbb{C}^n$ \cite{kr1} is a set of unitary operators $ E=\{\rho_g\in U(g): g\in G\}$ such that
    \begin{enumerate}
        \item $\rho (1)$ is the identity matrix, where $1$ denotes the identity element of the group $G$.
        \item ${\rm Tr}(\rho_g)=n\delta_{g,1}$.
        \item $\rho_g \rho_h=\omega(g,h)\rho_{gh}$, where $\omega(g,h)\in \C$ is a scalar.
    \end{enumerate}
    For such a set of operators the labelling group $G$ is called the index group of the corresponding NEB.
\end{definition}

Conditions (1) and (3) simply tell us that the representation $\rho$ is a projective representation. 
If we equip $M_n(\mathbb{C})$ with the inner product $\langle A,B \rangle:={\rm Tr}(A^*B)$ then conditions (1)-(3) ensure that $E$ is an orthonormal set since
\begin{align*}
    \langle \rho_g,\rho_h \rangle= {\rm Tr}(\rho_g^*\rho_h)
     = \omega(g^{-1}, g)^{-1}{\rm Tr}(\rho_{g^{-1}}\rho_h) 
    &=\omega(g^{-1},g)^{-1}\omega(g^{-1}, h){\rm Tr}(\rho_{g^{-1}h})\\
    &=\omega(g^{-1},g)^{-1}\omega(g^{-1}, h) n\delta_{g^{-1}h}.
\end{align*}

Comparing the dimension it follows that $E$ is an orthogonal basis of $M_n(\mathbb{C})$.
It is also easy to see that the associativity of the group $G$ implies that function $\omega: G\times G\rightarrow \C$ is a 2-cocycle.
If we normalise each $\pi_g$ so that $\rm{det}(\pi_g)=1$ then $\omega(g,h)$ becomes an $n$-th root of unity for any $g,h\in G$.

\begin{theorem}
    Let $\mathcal{E}=\{\pi(g);g\in G\}$ be a set of unitary matrices indexed by a finite group $G$. Then $\mathcal{E}$ is a NEB iff $\pi$ is a unitary faithful irreducible projective representation of order $|G|^{1/2}$.
\end{theorem}
\begin{proof}
See theorem 1, \cite{kr1}.    
\end{proof}

\subsection{NEB from Groups of Central Type:}

\begin{definition}
    A group H is called of central type if there exists an irreducible character $\chi$ such that $\chi(1)^2=|H:Z(H)|$, where $Z(H)$ is the center of the group.
\end{definition}
It is also equivalent to the existence of an irreducible character which vanishes outside $Z(H)$(cf. p. 28, \cite{Issac}).
Let $H$ be a group of central type with a unitary irreducible character $\chi$ verifying the above condition.
Suppose, $\mathfrak{X}_{\chi}$ is the unitary irreducible representation corresponding to the character $\chi$.
Recall that a representation $\mathfrak{X}$ is called faithful iff $\mathrm{Ker}(\mathfrak{X})=\{Id\}$.
If we define the kernel of a character $\chi$ as $\rm{Ker}\chi :=\{h\in H; \chi(h)=\chi(1)\}$ then it is well known that $\rm{Ker}(\mathfrak{X}_{\chi})=\rm{Ker}(\chi)$.
We take the quotient group $H/\text{Ker }\chi$ so that the natural representation  $\Tilde{\mathfrak{X}}_{\chi}(\Bar{h}):=\mathfrak{X}_{\chi}(h)$ of $H/{\rm Ker}\chi$ becomes faithful. 
\begin{lemma}\label{lm center qt.grp}
    Let $H$ be a group and $\chi$ is an irreducible character of $G$. 
    Then $Z(H/\rm{Ker}\chi)= Z(\chi)/\rm{Ker}\chi$, where $Z(\chi):=\{h\in H; |\chi(h)|=\chi(1)\}$ is called the center of the representation. 
    Moreover, $Z(\chi)/\rm{Ker}\chi$ is a cyclic group.
\end{lemma}
\begin{proof}
    Lemma 2.27, \cite{Issac}
\end{proof}
\begin{lemma}
    If $\chi$ is an irreducible character of the group $H$ then $\chi(1)^2\leq |H:Z(\chi)|$.
    Equality occurs if and only if $\chi$ vanishes on $G-Z(\chi)$.
\end{lemma}
\begin{proof}
    Corollary 2.30, \cite{Issac}.
\end{proof}

\begin{theorem}
    Let $H$ be a group of central type. Then the group\\ $(H/\rm{Ker}\chi)/(Z(H)/\rm{Ker}\chi)\cong H/Z(H)$ is an index group.
\end{theorem}

\begin{proof}
    Let $H$ be a group of central type with an irreducible character $\chi$ which satisfies $\chi(1)^2=|H:Z(H)|$ then $Z(\chi)=Z(H)$ [cf. p. 28, \cite{Issac}].
    Therefore we have that $Z(H/{\rm Ker}\chi) = Z(H)/{\rm Ker}\chi$ by Lemma \ref{lm center qt.grp}.
    For each $h\in (H/{\rm Ker}\chi)/(Z(H)/{\rm Ker} \chi)$ we choose a coset representative $\phi(h)$ in $H/{\rm Ker} \chi$. 
    Let us denote $\pi_h=\Tilde{\mathfrak{X}}_{\chi}(\phi(h))$. 
    Therefore we have 
    \[
        \pi_{h}\pi_{k}=\Tilde{\mathfrak{X}}(\phi(h))\Tilde{\mathfrak{X}}(\phi(k)) 
        = \Tilde{\mathfrak{X}}(\phi(h)\phi(k))
        =\Tilde{\mathfrak{X}}(\phi(hk)z_{h,k}),
    \]
    where $z_{h,k}\in Z(H)/\rm{Ker}\chi$.
    $\Tilde{\mathfrak{X}}_{\chi}(Z(H)/{\rm Ker} \chi)$ consists of scalar multiples of identity only. So we obtain
    $\pi_h.\pi_k= \Tilde{\mathfrak{X}}(\phi(hk))\Tilde{\mathfrak{X}}(z_{h,k})=\omega(h,k)\pi_{hk}$, where $\omega(h,k)\in \mathbb{C}$.
    Since the representation is irreducible, all the $\pi_h$ spans $M_n(\mathbb{C})$ for some $n\in \mathbb{N}$.
    Using isomorphism theorem we get $(H/\text{Ker }\chi)/(Z(H)/\text{Ker }\chi) \cong H/Z(H)$. As we know the character $\chi$ vanishes outside $Z(H)$, it follows that ${\rm Tr } (\pi_h)=0$ except at the identity. 
    So we find that $(H/\rm Ker \chi)/(Z(H)/\rm Ker \chi) \cong H/Z(H)$ is an index group if H is a group of central type.
\end{proof}

\textbf{Example:} 
We present here two examples to briefly clarify the previous result- one with an Abelian index group and another with a non-Abelian index group.
\begin{itemize}
    \item[i.] The group of unit quaternions $Q=\{\pm 1, \pm i, \pm j, \pm k\}$ (with multiplication as the group operation) has eight elements and five irreducible representations (up to equivalence), which we can choose as
\begin{eqnarray*}
\varepsilon &:& \varepsilon(i)=\varepsilon(j)=1, \\
\sigma_i &:& \sigma_i(i)=1,\quad \sigma_i(j)=-1, \\
\sigma_j &:& \sigma_j(i)=-1, \quad \sigma_j(j) =1, \\
\sigma_k &:& \sigma_k(i)=-1=\sigma_k(j), \\
\pi &:& \pi(i) = 
\left(
\begin{array}{cc}
0 & i \\
i & 0
\end{array}
\right), \quad
\pi(j)=
\left(
\begin{array}{cc}
0 & -1 \\
1 & 0
\end{array}
\right).
\end{eqnarray*}
For the character table we get
\[
\begin{array}{c||c|c|c|c|c|c|c|c||c}
& 1 & -1 & i & -i & j & -j & k & -k & \text{dim} \\ \hline
\chi_\varepsilon  = \varepsilon & 1 & 1 & 1 & 1 & 1 & 1 & 1 & 1 & 1 \\
\chi_{\sigma_i} = \sigma_i & 1 & 1 & 1 & 1 & -1 & -1 & -1 & -1 & 1 \\ 
\chi_{\sigma_j} = \sigma_j & 1 & 1 & -1 & -1 & 1 & 1 & -1 & -1 & 1 \\ 
\chi_{\sigma_k} = \sigma_k & 1 & 1 & -1 & -1 & -1 & -1 & 1 & 1 & 1 \\ 
\chi_\pi & 2 & -2 & 0 & 0 & 0 & 0 & 0 & 0 & 2
\end{array}
\]
We see that
\begin{eqnarray*}
{\rm ker}(\chi_\pi) &=& \{1\}, \\
Z(\chi_\pi) &=& \{1,-1\} = Z(Q).
\end{eqnarray*}
$Q$ is a group of central type: its center is $Z(Q)=\{-1,1\}$ and it has a $|Q/Z(Q)|^{1/2}$-dimensional irreducible representation. We have $Q/Z(Q) \cong \mathbb{Z}_2\times \mathbb{Z}_2$.

The 2-cocycle $\omega:\mathbb{Z}_2\times\mathbb{Z}_2\to \mathbb{T}$ is given by the relations
\[
\pi(g_1)\pi(g_2) = \omega(g_1,g_2)\pi(g_1g_2),
\]
for $g_1,g_2\in \Z_2$.
If we write $\mathbb{Z}_2\times\mathbb{Z}_2=\{(\pm 1,\pm 1)\}$ multiplicatively and choose
\begin{eqnarray*}
\pi(+1,+1) &=& \pi(1) = I_2, \\
\pi(+1,-1) &=& \pi(i) = \left(
\begin{array}{cc}
0 & i \\
i & 0
\end{array}
\right), \\
\pi(-1,+1) &=& \pi(j) = \left(
\begin{array}{cc}
0 & -1 \\
1 & 0
\end{array}
\right) \\
\pi(-1,-1) &=& \pi(k) = \left(
\begin{array}{cc}
i & 0 \\
0 & -i
\end{array}
\right).
\end{eqnarray*}

We get
\[
\begin{array}{c|cccc}
\omega & (+1,+1) & (+1,-1) & (-1,+1) & (-1,-1) \\ \hline 
(+1,+1) & 1 & 1 & 1 & 1 \\
(+1,-1) & 1 & -1 & 1 & -1 \\
(-1,+1) & 1 & -1 & -1 & 1 \\
(-1,-1) & 1 & 1 & -1 & -1 \\
\end{array}
\]

\item[ii.] Klappenecker and R\"otteler construceted \cite{kr1} an example of NEB corresponding to a non-commutative index group which we briefly mention here.
Consider the group $H_n$ for some $n\in \N$, generated by the composition of the maps
    \[
        \tau: x\mapsto x+1 (\mathrm{mod }\quad 2^n) \quad \text{and} \quad 
        \alpha: x\mapsto 5x (\mathrm{mod}\quad 2^n).
    \]
If $A:= \langle \tau\rangle$ and $B:= \langle \alpha \rangle$ then $H_n=A\rtimes B$.
\begin{theorem}
    The group $H_n$ is a group of central type of order $2^{2n-2}$ with cyclic center $Z(H_n)= \langle \tau^{2^{n-2}}\rangle $.
    The index group $H_n/Z(H_n)$ is non-Abelian for $n\ge 5$.
\end{theorem}
\begin{proof}
    See theorem 5, \cite{kr1}.
\end{proof}
Let $\phi: \Z/2^n\Z \rightarrow \C$ be a map defined by 
    \[
        \phi(x)= \exp{\left(\frac{2\pi i 5^x}{2^n}\right)}.
    \]
Then the diagonal matrix 
    \[
        {\pi}(\tau)= \mathrm{diag}(\phi(0), \phi(1),\ldots, \phi(2^{n-1}-1))
    \]
and the shift
    \[
        {\pi}(\alpha)= \begin{bmatrix}
            0 & 1 & 0 & \hdots & 0\\
            0 & 0 & 1 & \hdots & 0\\
            \vdots & & & \ddots & \vdots\\
            0 & 0 & 0 & \hdots & 1\\
            1 & 0 & 0 & \hdots & 0
        \end{bmatrix}
    \]
define an ordinary faithful irreducible representation of $H_n$.
The NEB corresponding to the index group $H_n/Z(H_n)$ is given by
    \[
        \mathcal{E}= \{\pi(\tau)^k\pi(\alpha)^l: 0\le k,l < 2^{n-2}-1\}.
    \]

\end{itemize}
For more examples of NEB and index groups, see the website
\url{https://people.engr.tamu.edu/andreas-klappenecker/ueb/ueb.html}

\subsection{NEB from Discrete Weyl Relation}
We define a bicharacter $\varkappa: \Z_n\times \Z_n \longrightarrow \C$ such that
\begin{itemize}
\item[(i)]
$|\varkappa( x,y)|=1$, for all $x,y\in G$,
\item[(ii)]
Symmetry: $\varkappa( x,y) = \varkappa( y,x)$, for all $x,y\in G$,
\item[(iii)]
Non-degeneracy: $\varkappa( x,y)=1$ for all $y\in G$ iff $x=0$,
\item[(iv)]
Character: $\varkappa( x,y+z)=\varkappa(x,y)\cdot \varkappa( x,z)$.
\end{itemize} 
For example we can take
\[
\varkappa(k,\ell) = \exp\left(\frac{2\pi i k \ell}{n}\right), \qquad k,\ell\in \mathbb{Z}_n.
\]
We fix an orthonormal basis, $\{|x\rangle; x\in \Z_n\}$ of $\C^n$.
We define two unitary representation $U, V$ of $\Z_n$ on $\C^n$, by the relation
\begin{eqnarray*}
U_a |x\rangle &=& |x+a\rangle, \\
V_a|x\rangle &=& \varkappa(a,x) |x\rangle,
\end{eqnarray*}
The operators $U_a,V_b$, $a,b\in Z_n$, satisfy the \textit{Weyl commutation relations}
\[
U_aU_b = U_{a+b}, \qquad V_aV_b=V_{a+b}, \qquad V_b U_a = \varkappa(a,b) U_aV_b
\]
for $a,b\in \Z_n$.

We define the \textit{Weyl operators}
\[ 
W_{a,b} = U_aV_b
\] 
for $a,b\in \Z_n$ (see p. 212, \cite{Wat}).
The matrix coefficients of a Weyl operator $W_{a,b}$ w.r.t.\ to the basis $\{|x\rangle;x\in \Z_n\}$ are given by
\[
\langle y| W_{a,b}| x \rangle = \varkappa(b,x) \delta_{y,x+a}, \qquad x,y\in \Z_n, 
\]
or, equivalently,
\begin{equation} \label{eq-W-coeff}
W_{a,b} = \sum_{x\in G} \varkappa(b,x)\, |x+a\rangle\langle x|.
\end{equation}
These Weyl operators $\{W_{a,b}; (a,b)\in Z_n\times \Z_n\}$ form a nice error basis and $\Z_n\times \Z_n$ becomes its index group.
To check that we can verify some straightforward relations
\begin{itemize}
    \item [(i)] (Unitary)$W_{a,b}^* = W_{a,b}^{-1} = \varkappa(a,b) W_{-a,-b}$.
    \item [(ii)](Projective representations) $W_{a,b}W_{x,y}=\varkappa(b,x)W_{a+x.b+y}$.
\end{itemize}
If we compute the trace of the Weyl operators with respect to the Hilbert-Schmidt inner product on $M_n$
\begin{equation*}
{\rm Tr} (W_{a,b}) = \sum_{x\in G} \langle x| \underbrace{U_aV_b |x\rangle}_{=\varkappa(b,x)|x+a\rangle} = \left\{
\begin{array}{ll} n & \mbox{ if } (a,b)=(0,0), \\ 0 & \mbox{ else.}\end{array}
\right.
\end{equation*}
Here we used the orthogonality of the characters,
\[
\sum_{x\in G} \underbrace{\overline{\varkappa(a,x)} \varkappa(b,x)}_{=\varkappa(b-a,x)} = \#G\, \delta_{a,b}. 
\]

\section{Convenient Basis of $\mathrm{Lin}(M_n(\C))$}
For a pair of matrices $A,B\in M_n$ we define a linear map $T_{A,B}:M_n\to M_n$ by
\[
T_{A,B}(X) = AXB^*, \qquad \text{ for } X\in M_n.
\]
\begin{proposition}\label{isomorphism}
The map $M_n(\C)\times M_n(\C)\ni (A,B)\mapsto T_{A,B}\in \mathrm{Lin}(M_n(\C))$ extends to an unique isomorphism of *-algebras $T:M_n\otimes M_n^*\ni A\otimes B^* \to T_{A, B}\in \mathrm{Lin}(M_n,M_n)$, which is also an isomorphism of Hilbert spaces.
\end{proposition}

The above proposition shows that if $\{B_i;1\le i \le n^2\}$ is a basis of $M_n$ and we define $T_{ij}(X)=B_iXB_j^*$ for any $X\in M_n$ then $\{T_{ij}\in \rm{Lin}(M_n);1\le i,j \le n^2\}$ is a basis of $\rm{Lin}(M_n)$. Taking a NEB as a basis of $M_n$ has the added advantage that in that case $T_{ij}$ also becomes NEB in $\rm{Lin}(M_n)$.
\begin{lemma}
    Let $\{\frac{1}{\sqrt{n}}\pi_g;g\in G\}$ be a NEB of $M_n$ corresponding to an index group $G$. Define the linear map $T_{x,y}:M_n(\C)\rightarrow M_n(\C)$
    \[
        T_{x,y}(X):= \pi_xX\pi_y^*
    \]
    Then the $\{\frac{1}{n}T_{x,y}; x,y\in G\}$ is a NEB of $\textrm{Lin}(M_n)$ with index group $G\times G$ and 2-cocycle $\omega^L:(G\times G)\times (G\times G)\rightarrow \mathbb{T}$ given by
    \[
        \omega^L((x',y'),(x,y))= \frac{\omega(x',x)}{\omega(y',y)}.
    \]
\end{lemma}
\begin{proof}
    It is trivial to check that $T_{1,1}=\textrm{Id}$ where $1$ is the identity element of $G$. 
    To check the trace condition we compute 
    \begin{gather*}
        \mathrm{Tr}_{\mathrm{Lin}(M_n)}(T_{x,y})= \sum \Big\langle |i\rangle\langle j|, T_{xy}|i\rangle\langle j|\Big\rangle
        = \sum \mathrm{Tr}_{M_n}\Big(|j\rangle\langle i|\pi_x|i\rangle\langle j|\pi_y^* \Big)\\
        = \sum \langle i|\pi_x|i\rangle \langle j|\pi_y^*|j\rangle
        = \mathrm{Tr}(\pi_x)\mathrm{Tr}(\pi_y^*)= n^2\delta_{1,x}\delta_{1,y},
    \end{gather*}
    which shows that trace of each operator $T_{x,y}$ is zero except for the identity.
    For any $X\in M_n$ and $x,y,x',y'\in G$ we see that 
    \[
    T_{x',y'}\circ T_{x,y}(X)= \pi_{x'}\pi_xX\pi_y^*\pi_{y'}^*
    = \frac{\omega(x',x)}{\omega(y',y)}\pi_{x'x}X\pi_{y'y}^*
    = \frac{\omega(x',x)}{\omega(y',y)}T_{x'x,y'y}(X),
    \]
    which proves the claim about the two cocycle $\omega^L$.
\end{proof}
The next proposition follows immediately since any NEB forms an ONB of the associated space of linear maps.
\begin{proposition}
The set $\{\frac{1}{n} T_{x,y}; x,y\in G\}$ forms an orthonormal basis of $\mathrm{Lin}(M_n(\mathbb{C}))$ with respect to Hilbert-Schmidt inner product.
\end{proposition}

Let $\alpha\in \mathrm{Lin}(M_n(\C))$ and $\{B_i;1\le i\le n^2\}$ be a basis of $M_n(\C)$. 
Since $T_{ij}$(defined above ) forms a basis of $\rm{Lin}(M_n)(\C)$ we can decompose $\alpha$ as
\begin{equation}\label{D_alpha gen}
    \alpha(X)= \sum_{1\le i,j\le n^2} D_{\alpha}(i,j)B_iXB_j^*.
\end{equation}
In particular if we take a NEB $\{\frac{1}{\sqrt{n}}\pi_x;x\in G\}$ as basis of $M_n(\C)$ we can explicitly compute the coefficient $D_{\alpha}$. 
\begin{equation}\label{D_alpha_def}
    \alpha(X)= \sum_{x,y} D_{\alpha}(x,y)T_{x,y}(X) =\sum_{x,y} D_{\alpha}(x,y)\pi_x X \pi_y^*.
\end{equation}
for all $X \in M_N(\mathbb{C})$.
Using the orthonormality of the basis $T_{x,y}$ and NEB $\{\frac{1}{\sqrt{n}}\pi_g;g\in G\}$ of $M_n(\C)$, we have $D_{\alpha}(x,y)= \frac{1}{n}\langle T_{x,y},\alpha\rangle_{\rm{Lin}(M_N)} $ i.e.
\begin{eqnarray}
D_{\alpha}(x,y) &=& \frac{1}{n^2}{\rm Tr}(T_{x,y}^{\dag}\alpha)\nonumber\\
&=& \frac{1}{n} \sum_{g \in G}\langle T_{x,y}(\pi_g), \alpha(\pi_g)\rangle\nonumber\\
&=& \frac{1}{n^2} \sum_{g \in G} {\rm Tr}(\pi_y\pi_g^*\pi_x^*\alpha(\pi_g)).\label{D_alpha}
\end{eqnarray}
Here $T^{\dag}$ denotes the involution applied on $T$ w.r.t the Hilbert-Schmidt inner product on $\mathrm{Lin}(M_n(\C))$.
If we use the Weyl operators as the nice error basis for defining $T_{x,y}$ i.e. $T_{x,y}(X)= W_xXW_y^*$ for $x,y\in \Z_n\times \Z_n$ then we can compute the coefficient $D_{\alpha}$ in \eqref{D_alpha}, using $\{|a\rangle \langle b|; a,b\in \Z_n\}$ as a o.n.b of $M_n(\C)$
\begin{eqnarray}\label{D_alpha weyl}
D_\alpha(x,y) &=& \frac{1}{n}\sum_{a,b\in G}{\rm Tr}\Big(W_y|b\rangle\langle a| W_x^* \alpha\big(|a\rangle\langle b|\big)\Big) \nonumber \\
&=& \frac{1}{n} \sum_{a,b\in G} \frac{\varkappa(y_2,b)}{\varkappa(x_2,a)}\Big\langle a+x_1\Big| \alpha\big(|a\rangle\langle b|\big)\Big|b+y_1\Big\rangle,
\label{eq-D-coeff-ab}
\end{eqnarray}
for $x=(x_1,x_2), y=(y_1,y_2)$.

\begin{lemma}\label{comp-coef}
Let $\alpha, \beta \in \mathrm{Lin}(M_n(\C))$ with coefficients $(D_{\alpha}(x,y))_{x,y\in G}$, as defined in the equation \eqref{D_alpha_def}. Then the coefficient of their composition $\alpha \circ \beta$ are given by
\[
D_{\alpha\circ \beta}(x,y)= \sum_{p,q\in G}\omega(p,xp^{-1})\overline{\omega(q, yq^{-1})}D_{\alpha}(p,q)D_{\beta}(p^{-1}x,q^{-1}y).
\]
\end{lemma}
\begin{proof}
    We have
    \begin{eqnarray*}
        \alpha\circ \beta(X)&=& \sum_{p,q\in G} D_{\alpha}(p,q)\pi_p \left( \sum_{p',q'\in G}D_{\beta}(p',q')\pi_{p'}X\pi_{q'}^* \right)\pi_q^*\\
        &=& \sum_{p,p',q,q'\in G}\omega(p,p')\overline{\omega(q,q')} D_{\alpha}(p,q)D_{\beta}(p',q')\pi_{pp'}X\pi_{qq'}^*\\
        &=& \sum_{x,y\in G} \underbrace{\left( \sum_{p,q\in G}\omega(p,p^{-1}x)\overline{\omega(q,q^{-1}y)}D_{\alpha}(p,q)D_{\beta}(p^{-1}x,q^{-1}y)\right)}_{D_{\alpha \circ \beta}}\pi_xX\pi_y^*,
    \end{eqnarray*}
    which completes the proof.
\end{proof}

We can define two different involutions on $\mathrm{Lin}(M_n(\C))$. The first comes from the Hilbert-Schmidt inner product and is characterised by the condition
\[
\langle X, \alpha(Y)\rangle = \langle \alpha^\dag(X),Y\rangle
\]
for all $X,Y\in M_n(\mathbb{C})$.

The second is inherited from the involution in $M_n(\mathbb{C})$ and defined by $\alpha^\# (X) = \alpha(X^*)^*$.

Both involutions are conjugate linear, but only the first is anti-multiplicative, whereas the second is multiplicative, i.e., we have
\[
(\alpha\circ\beta)^\dag = \beta^\dag \circ \alpha^\dag, \qquad (\alpha\circ\beta)^\# = \alpha^\# \circ \beta^\#
\]
for $\alpha,\beta\in \rm{Lin}(M_n(\C))$.

\begin{proposition}\label{involution}
We have
\[
D_{\alpha^\dag} (x,y)= \frac{\omega(x,x^{-1})}{\omega(y,y^{-1})} \overline{D_\alpha(x^{-1},y^{-1})}
\]
and
\[
D_{\alpha^\#}(x,y) = \overline{D_\alpha(y,x)}
\]
for $x,y\in G$.
\end{proposition}

\begin{proof}
It is easy to see that for any $x,y\in G$ we have $T_{x,y}^{\dag}= \frac{\omega(x,x^{-1})}{\omega(y,y^{-1})}T_{x^{-1},y^{-1}}$.
Applying the involution $\dag$ on the decomposition of $\alpha$ in \eqref{D_alpha_def}
\[
    \alpha^{\dag}= \sum_{x,y\in G} \overline{D_{\alpha}(x,y)}T_{x,y}^{\dag}
    = \sum_{x,y\in G} \frac{\omega(x,x^{-1})}{\omega(y,y^{-1})}\overline{D_{\alpha}(x,y)}T_{x^{-1},y^{-1}}
\]
the first claim follows.
Similarly, we can trivially check that $T_{x,y}^\#= T_{y,x}$ for any $x,y\in G$.
Then the second claim follows by applying $\#$ on the decomposition \eqref{D_alpha_def} 
\[
    \alpha^{\#}= \sum_{x,y\in G} \overline{D_{\alpha}(x,y)}T_{x,y}^\#
               = \sum_{x,y\in G} \overline{D_{\alpha}(x,y)}T_{y,x} 
\]
\end{proof}

\subsection{Examples}\label{example}
Here we compute the kernel or the $n^2\times n^2$ matrix $D_{\alpha}$ corresponding to different positive maps $\alpha \in \mathrm{Lin}(M_n(\C))$ which are important in quantum information. \noindent\\
\textbf{Identity map:} The identity map corresponds to the kernel $D_{\text{Id}}(x,y)= \delta_{1,x}\delta_{1,y}$ for $x,y\in G$ as we can write
\[
    X= D_{1,1}\pi_1X\pi_1^*= \sum_{x,y\in G} \delta_{1,x} \delta_{1,y} \pi_xX\pi_y^*.
\]
\\
\textbf{Depolarising Channel:} Let $P\in\mathrm{Lin}(M_n(\C))$ be the diagonal sum $P= \sum_{g\in G} T_{g,g}$. For any $h\in G$ and $X\in M_n(\C)$ we have
\[
\pi_hP(X)=\sum_g \omega(h,g)\pi_{hg}X\pi_g^* \quad \text{and} \quad 
P(X)\pi_h= \sum_g \frac{\omega(h,h^{-1})}{\omega(h^{-1},g)}\pi_gX\pi_{h^{-1}g}^*.
\]
After a change of variable we can write $P(X)\pi_h= \sum_g\frac{\omega(h,h^{-1})}{\omega(h^{-1},hg)}\pi_{hg}X\pi_g^*$. Using the defintion of 2-cocycle 
\[
\frac{\omega(h,h^{-1})}{\omega(h^{-1},hg)}= \frac{\omega(h,h^{-1})} {\omega(h,h^{-1})\omega(1,g)\overline{\omega(h,g)}}= \omega(h,g).
\]
Thus we see that $P(X)$ commutes with every basis element $\pi_h$ on $M_n$. So we conclude that $P(X) = cI_n$ for some $c\in \C$. 
Computing the trace of both sides we find that $c= \rm{Tr}(X)$.
Therefore we see that the map $P$, defined as the diagonal sum of the operators $T_{g,g}$s, is actually the depolarising channel which corresponds to the identity matrix $D_P(x,y)= \delta_{x,y}$ where $x,y\in G$.\\

\textbf{Transposition:}
Let $T\in \mathrm{Lin}(M_n(\C))$ be the transposition map given by $X\mapsto X^t$.
We compute the $n^2\times n^2$ matrix $D_T$ corresponding to the transposition map. For any$x,y\in G$ we have
\begin{eqnarray*}
    D_T(x,y)= \langle T_{x,y}|T\rangle &=& \sum_{1\le i,j \le n} \mathrm{Tr} \Big( (T_{x,y}|i\rangle\langle j|)^* T|i\rangle\langle j| \Big)\\
    &=& \sum_{1\le i,j\le n} \mathrm{Tr}\Big( \pi_y|j\rangle\langle i|\pi_x^*|j\rangle\langle i| \Big)\\
    &=& \sum_{1\le i,j\le n} \langle i|\pi_x^*|j\rangle \mathrm{Tr}\Big(\pi_y|j\rangle \langle i|\Big)\\
    &=& \sum_{1\le i,j\le n} \langle i|\pi_x^*|j\rangle \langle i|\pi_y|j\rangle =\mathrm{Tr}(\overline{\pi}_x\pi_y)
\end{eqnarray*}
In particular, if we take the Weyl operators $\{W_{a,b}; a,b\in \Z_n\}$ as the chosen NEB and $\{|i\rangle; i\in \Z_n\}$ as the standard basis of $\C^n$ then
\begin{gather*}
    D_{T}((a,b),(c,d))= \sum_{i,j\in \Z_n}\langle i|W_{a,b}^*|j\rangle \langle i|W_{c,d}|j\rangle
    = \sum_{i,j\in \Z_n} {\chi(a,b)}\chi(-b,j)\chi(d,j)\delta_{i,j-a}\delta_{i,c+j}.
\end{gather*}
So we have
\[ D_T{((a,b),(c,d))}= 
\left\{
\begin{array}{cc}
     \sum_{i\in \Z_n}{\chi(a,b)}\chi(d-b,i+a) & \text{ if } c=-a,  \\
     0 & \text{otherwise.} 
\end{array}\right.
\]
\\
\textbf{Conditional Expectation onto Diagonal:} 
Consider the linear map $C: M_n(\C)\rightarrow M_n(\C)$, $C(X)= (\delta_{ij}x_{ij})_{1\le i,j\le n}$ for $X=(x_{ij})_{1\le i,j\le n}\in M_n(\C)$.
This map is a conditional expectation onto the *-subalgebra of diagonal matrices with respect to the standard basis. We compute the coefficient matrix $D_C$ of the map $C$. For any $x,y\in G$
\begin{eqnarray*}
    D_C(x,y)= \langle T_{x,y}|C\rangle &=& \sum_{1\le i,j\le n} \mathrm{Tr}\Big((T_{x,y}|i\rangle\langle j|)^*C|i\rangle\langle j|\Big)\\
    &=& \sum_{1\le i,j\le n}\delta_{ij}\mathrm{Tr}(\pi_y|j\rangle\langle i|\pi_x^*|i\rangle\langle j|)\\
    &=& \sum_{1\le j\le n} \delta_{ij}\langle i|\pi_x^*|i\rangle \langle j|\pi_y|j\rangle = \mathrm{Tr}(C(\pi_x^*)C(\pi_y)).
\end{eqnarray*}
In particular, taking the Weyl operators as NEB gives 
\[
 D_C((a,b),(c,d)= \left\{
\begin{array}{cc}
  \sum_{j\in \Z_n}\chi(d-b,j) & \text{ if } c=a=0,\\
  0 & \text{otherwise.}
\end{array}\right.
\]

\section{Correspondence between Choi matrix $C_{\alpha}$ and $D_{\alpha}$}
Recall that the Choi-Jamio{\l}kowski matrix of a map $\alpha\in \mathrm{Lin}(M_n(\mathbb{C}))$ is the $n^2\times n^2$-matrix defined by
\[
C_\alpha = \sum_{j,k=1}^n E_{jk}\otimes \alpha(E_{jk}) \in M_n(\mathbb{C})\otimes M_n(\mathbb{C}) \cong M_{n^2}(\mathbb{C}).
\]
It is known that $\alpha\in \mathrm{Lin}(M_n(C))$ is completely positive (CP) iff $C_\alpha$ is positive. Furthermore, $\alpha\in \mathrm{Lin}(M_n(\C))$ is $k$-positive if and only if
\[
\langle v, C_\alpha v\rangle \ge 0
\]
for all $v\in \mathbb{C}^n\otimes \mathbb{C}^n$ with Schmidt rank not more than $k$.
For any completely positive map $\alpha$ we have the Kraus decompositon 
\[
    \alpha = \sum_{j=1}^r\mathrm{Ad}_{L_j}
\]
for some matrices $L_j\in M_n(\C)$, where for any $X\in M_n(\C)$ the conjugate map $\mathrm{Ad}_{L_j}$ is given by $\mathrm{Ad}_{L_j}(X)=L_jXL_j^*$.
We call $\alpha$ $1$-super positive or entanglement breaking iff $\text{rank}(L_j)= 1$ for any $j$.
It is called completely co-positive iff $T\circ \alpha$ is completely positive, where $T$ is the transposition map,
see \cite{ra07,ssz}.

We can switch from $C_\alpha$ to $D_\alpha$ by a change of basis.

\begin{proposition} If $T_{x,y}$ is defined with respect to the Weyl operators then the Choi-Jamio{\l}kowski matrix of $T_{x,y}$ is given by
\[
C_{T_{x,y}} (v,w) = \frac{\varkappa(x_2,v_1)}{\varkappa(y_2,w_1)}\delta_{v_1+x_1,v_2}\delta_{w_1+y_1,w_2},
\]
for $x=(x_1,x_2),y=(y_1,y_2),v=(v_1,v_2),w=(w_1,w_2)\in \Z_n\times \Z_n$. More generally, if $\alpha$ is of the form $\alpha = \sum_{x,y\in \Z_n\times Z_n} D_\alpha(x,y) T_{x,y}$, then its Choi-Jamio{\l}kowski matrix is given by
\[
C_\alpha(v,w) =\sum_{x_2,y_2\in G} \frac{\varkappa(x_2,v_1)}{\varkappa(y_2,w_1)} D_\alpha\big((v_2-v_1,x_2),(w_2-w_1,y_2)\big),
\]
for $v,w\in \Z_n\times \Z_n$.
Conversely, the coefficients from the equation \eqref{D_alpha weyl} can be computed from the Choi-Jamio{\l}kowski matrix via
\[
D_\alpha(x,y) = \frac{1}{n} \sum_{a,b\in G} \frac{\varkappa(y_2,b)}{\varkappa(x_2,a)} C_\alpha\big((a,a+x_1),(b,b+y_1)\big)
\]
for $x,y\in \Z_n\times \Z_n$.
\end{proposition}

\begin{proof}
Using Formula \eqref{eq-W-coeff}, we get
\[
W_y^* = \sum_{d\in \Z_n} \frac{1}{\varkappa(y_2,d)} |d\rangle\langle d+y_1|
\]
and
\[
T_{x,y}\big(|a\rangle\langle b|\big) = \frac{\varkappa(x_2,a)}{\varkappa(y_2,b)} |a+x_1\rangle\langle b+y_1|,
\]
for $a,b\in \Z_n$, $x,y\in \Z_n\times \Z_n$. So, if we choose $\{|a\rangle;a\in \Z_n\}$ as a basis of $\mathbb{C}^n$, we can write the corresponding matrix units as $|a\rangle\langle b|$, $a,b\in \Z_n$, and we get
\begin{eqnarray*}
C_{T_{x,y}} &=& \sum_{a,b\in \Z_n} |a\rangle\langle b| \otimes T_{x,y}\big(|a\rangle\langle b|\big) \\
&=& \sum_{a,b\in \Z_n} \frac{\varkappa(x_2,a)}{\varkappa(y_2,b)} |(a,a+x_1)\rangle\langle (b,b+y_1)|,
\end{eqnarray*}
which proves the first claim of the proposition.

For $\alpha = \sum_{x,y\in \Z_n\times \Z_n} D_\alpha(x,y) T_{x,y}$, this yields
\[
C_\alpha = \sum_{x,y\in \Z_n\times \Z_n} \sum_{a,b\in \Z_n} D_\alpha(x,y) \frac{\varkappa(x_2,a)}{\varkappa(y_2,b)} |(a,a+x_1)\rangle\langle (b,b+y_1)|
\]
or
\[
C_\alpha(v,w) =\sum_{x_2,y_2\in \Z_n} \frac{\varkappa(x_2,v_1)}{\varkappa(y_2,w_1)} D_\alpha\big((v_2-v_1,x_2),(w_2-w_1,y_2)\big).
\]
For the converse we use the equation \eqref{eq-D-coeff-ab},
\begin{eqnarray*}
D_\alpha(x,y) &=& \frac{1}{N}  \sum_{a,b\in G} \frac{\varkappa(y_2,b)}{\varkappa(x_2,a)}\Big\langle a+x_1\Big| \alpha\big(|a\rangle\langle b|\big)\Big|b+y_1\Big\rangle \\
&=& \frac{1}{N}  \sum_{a,b\in \Z_n} \frac{\varkappa(y_2,b)}{\varkappa(x_2,a)}\Big\langle a+x_1\Big| \left(\sum_{g,h\in \Z_n} C_\alpha\big((a,g),(b,h)\big) |g\rangle\langle h|\right) \Big|b+y_1\Big\rangle \\
&=& \sum_{a,b\in \Z_n} \frac{\varkappa(y_2,b)}{\varkappa(x_2,a)} C_\alpha\big((a,a+x_1),(b,b+y_1)\big),
\end{eqnarray*}
where we used the identity
\[
\alpha\big(|a\rangle\langle b|\big) = \sum_{g,h\in G} C_\alpha\big((a,g),(b,h)\big)  |g\rangle\langle h|.
\]
\end{proof}

%%%%%%%%%%%%%%%%%%%%%%%%%%%%%%%%%%%%%%%%%%%%%%%%%%%%%%%%

\section{Characterisation of positive and completely positive maps}
\begin{theorem}
Let $\{B_x\}_{x=1,2...n^2}$ be a basis of $M_n(\C)$. 
Consider a linear map $\alpha \in L(M_n(\C))$ of form $\alpha(X)=\sum_{x,y=1}^{n^2}D_{\alpha}(x,y)B_xXB_y^*$. 
Then $\alpha$ is 
\begin{enumerate}
    \item[i.] Hermitianity preserving if and only if $D_{\alpha}$ is Hermitian. 
    \item[ii.] positive if and only if for any $v, w\in \C^n$,
\[
\langle v\otimes w, \Tilde{\alpha}(v\otimes w)\rangle \geq 0
\]
\end{enumerate}
where $\Tilde{\alpha}= \tau \circ \sum_{x,y=1}^{n^2}D_{\alpha}(x,y) ( B_x\otimes B_y^*)$ and $\tau(u\otimes v)= v\otimes u$ is the flip operator.
\end{theorem}
\begin{proof}
$\alpha \in \rm{Lin}(M_n(\C))$ is Hermitianity preserving iff $\alpha(X^*)^*=\alpha(X)$ i.e. $\alpha^{\#}=\alpha$.
Comparing the coefficient matrix of both sides the first claim follows directly from the proposition \ref{involution}.

On the other hand, $\alpha$ positive if and only if it maps rank one projection to positive operators. i.e. for all $v,w \in \C$
\begin{align*}
    0\leq \left \langle v, \alpha(|u\rangle \langle u|)v\right\rangle &= \left\langle v, \sum D_{\alpha}(x,y)B_x|u\rangle \langle u|B_y^*v\right\rangle\\
    &= \sum D_{\alpha}(x,y)\langle v, B_xu\rangle \langle u, B^*_y v\rangle\\
    &= \left \langle u\otimes v, \tau \circ \sum D_{\alpha}(x,y) B_x\otimes B_y^* (u \otimes v) \right \rangle\\
    &= \langle u\otimes v, \Tilde{\alpha} (u \otimes v) \rangle.
\end{align*}
\end{proof}

\begin{theorem}\label{thm CP}
    A linear map $\alpha\in \mathrm{Lin}(M_n)$ is a completely positive map with Kraus rank $r$ then if and only if the corresponding coefficient matrix $D_{\alpha}\in M_{n^2}$ as defined in \eqref{D_alpha gen}, is positive semi-definite of rank $r$.
\end{theorem}
\begin{proof}
    If $\alpha \in L(M_n)$ is CP with Kraus rank $r$ then there exist $\{L_j\in M_n; 1\le j\le r\}$ such that $\alpha$ can be written as Kraus decomposition 
    \[
        \alpha = \sum_{j=1}^r Ad_{L_j},
    \]
    where $Ad_{L_j}$ is the conjugate map given by $Ad_{L_j}(X)= L_jXL_j^*$ for any matrix $X\in M_n$.
    Since the map $L(M_n)\ni \alpha \mapsto D_{\alpha}\in M_{n^2}$ is a linear isomorshism, we have $D_{\alpha}= \sum_{j=1}^r D_{Ad_{L_j}}$.
    We can expand each $L_j$ with respect to the NEB $\{B_x; x\in G\}$ and write $L_j= \sum_z l_j(z)B_z$. 
    Therefore we get 
    \begin{gather*}
        Ad_{L_j}(X)= L_jXL_j^*= \left(\sum_{z\in G}l_j(z)B_z\right)X\left(\sum_{z'\in G}l_j(z')B_{z'}\right)^*\\
        = \sum_{z,z'\in G}l_j(z)\overline{l_j(z')}B_zXB_{z'}^*
    \end{gather*}
    We find that $D_{Ad_{L_j}}$ is a rank one operator given by $$D_{Ad_{L_j}}= |l_j\rangle \langle l_j|$$ where $l_j=(l_j(1), l_j(2),\ldots, l_j(n^2))^t$ is a vector in $\C^{n^2}$. 
    Thus $D_{\alpha} = \sum_{j=1}^r |l_j\rangle \langle l_j|$ is a positive semi-definite operator of rank $r$.

    Conversely, assume that $D_{\alpha}$ is positive semi-definite with rank $r$.
    So there exists vectors $v_1, v_2, \ldots , v_r\in \C^{n^2}$ such that $D_{\alpha}= \sum_{j=1}^r |v_j\rangle \langle v_j|$.  
    If we denote $\{|x\rangle ; x\in G\}$ the standard basis on $\C^{n^2}$ then
    \[
        D_{\alpha}(x,y)= \sum_{j=1}^r\left\langle x|v_j\rangle \langle v_j|y\right \rangle = \sum_{j=1}^r \langle x|v_j\rangle \overline{\langle y|v_j\rangle}.
    \]
     Therefore we can write the equation \eqref{D_alpha gen} as
     \begin{eqnarray*}
         \alpha(X)&=& \sum_{x,y\in G}\sum_{j=1}^r\langle x|v_j\rangle \overline{\langle y|v_j\rangle}  B_xXB_y^* \\
         &=& \sum_{x,y\in G}\sum_{j=1}^r \big( \langle x|v_j\rangle B_x\big)X\big(\langle y|v_j\rangle  B_y\big)^*\\
         &=& \sum_{j=1}^r \left( \sum_x \langle x|v_j\rangle B_x\right)X \left(\sum_y \langle y|v_j\rangle B_y\right)^*.
    \end{eqnarray*}
    If we denote $L_j= \sum_{x\in G}\langle x|v_j\rangle  B_x$ then we get $\alpha(X)=\sum_{j=1}^r L_jXL_j^*$, which shows that $\alpha$ is completely positive with Kraus rank $r$.
\end{proof}

We remark here that the similar result was obtained by Poluikis and Hill from a different approach \cite{polhill}.
But our approach has the advantange that coefficient matrix corresponding to the composition of two linear maps becomes the convolution type product of their individual coefficient matrices.
This composition law can be used to characterise the completely co-positive maps too.
\begin{corollary}
    A linear map $\alpha\in \mathrm{Lin}(M_n(\C))$ is completely co-positive iff the convolution product 
    \[
    \sum_{p,q\in \Z_n\times \Z_n} \frac{\chi(p,x-p)}{\chi(q,y-q)}\textrm{Tr} 
    (\overline{W}_pW_q)D_{\alpha}(x-p,y-q)
    \]
    is positive semidefinite.
\end{corollary}
\begin{proof}
    Using the Theorem \ref{thm CP} a linear map $\alpha\in \mathrm{Lin}(M_n(\C))$ is co-CP iff in the decomposition $T\circ \alpha$ w.r.t the Weyl operators 
    \[
    T\circ \alpha(X)= \sum_{x,y\in \Z_n\times Z_n}D_{T\circ \alpha}W_xXW_y^*,
    \]
    the coefficient matrix $D_{T\circ \alpha}$ is positive semi-definite.
    We use the Lemma \ref{comp-coef} and the coefficients we have found for transposition in example \ref{example} to compute the coefficient $D_{t\circ \alpha}$ which complete the claim. 
\end{proof}

\begin{corollary}
A linear map $\alpha\in \mathrm{Lin}(M_2)$ is $1$-super positive iff $D_{\alpha}= \sum_{j=1}^r |l_j\rangle \langle l_j|$ where $l_j=(l_j(1),\ldots , l_j(4))^t$ is a vector in $\C^4$ satisfying $l_j(1)^2=\sum_{k=2}^4 l_j(k)^2$. 
\end{corollary}
\begin{proof}
    Since $\alpha$ is $1$-super positive in $M_2$ there exists matrices $L_1,L_2,\ldots, L_r\in M_2$ of rank $1$ such that $\alpha= \sum_{j}^rAd_{L_j}$.
    We can decompose each $L_j$ w.r.t the Pauli basis
    \[
\begin{array}{cccc}
     \sigma_1= \frac{1}{2}
     \begin{bmatrix}
         1 & 0\\
         0 & 1
     \end{bmatrix}, & 
     \sigma_2 =\frac{1}{2}
     \begin{bmatrix}
         0 & 1\\
         1 & 0
     \end{bmatrix}, &
     \sigma_3= \frac{1}{2}
     \begin{bmatrix}
         0 & -i\\
         i & 0
     \end{bmatrix}, &
     \sigma_4= \frac{1}{2}
     \begin{bmatrix}
         1 & 0\\
         0 & -1
     \end{bmatrix}.   
 \end{array}
 \]
 to obtain $L_j=\sum_{i=1}^4 l_j(k)\sigma_k$. After this decomposition $L_j$ has the form
 \[ \frac{1}{2}
\begin{bmatrix}
    l_j(1)+l_j(4) & l_j(2)-il_j(3)\\
    l_j(2)+il_j(3) & l_j(1)-l_j(4)
\end{bmatrix}.
\]
$L_j$ has rank $1$ iff $\textrm{det} (L_j)=0$ i.e. $l_j(1)^2=l_j(2)^2+l_j(3)^2+l_j(4)^2$.
We have already seen that in the proof of \ref{thm CP} that $D_{Ad_{L_j}}=|l_j\rangle\langle l_j|$,
which completes the the claim.
\end{proof}

\begin{proposition}
A linear map $\alpha: M_n(\C)\longrightarrow M_n(\C)$ 

    \item[(a)]is trace preserving if and only if
            \begin{equation*}
            \sum_{x\in G} \omega(x,g)D_{\alpha}(x,xg)= \delta_{1,g}
            \end{equation*}
            for all $g\in G$.
    \item[(b)] is unit preserving if and only if 
            \begin{equation*}
                \sum_{x\in G} \frac{\omega(x,x^{-1}z)}{\omega(z^{-1}x,x^{-1}z)}D_{\alpha}(x,z^{-1}x) = \delta_{1,z}
            \end{equation*}
            for all $z\in G$.
\end{proposition}

\begin{proof}
    \item[(a)] Since $\{\pi_g;g\in G\}$ forms a basis of $M_n$ the map $\alpha$ is trace preserving iff ${\rm Tr}(\alpha(\pi_g))={\rm Tr}(\pi_g$) for all $g\in G$.
    Now
    \[
    {\rm Tr}\alpha(\pi_g)=\sum_{x,y}D_{\alpha}(x,y)\frac{\omega(x,g)\omega(xg,y^{-1})}{\omega(y^{-1},y)}{\rm Tr}(\pi_{xgy^{-1}})
    \]
    Substituting $xg=y$ we get 
    \[
    \sum_{x\in G} \omega(x,g)D_{\alpha}(x,xg)= {\rm Tr}(\pi_g)= \delta_{1,g}
    \]
    for all $g\in G$.
    
    \item[(b)]By definition $\alpha$ is unit preserving iff $\alpha(I_n)=I_n$.
    \begin{align*}
        \alpha(I_n)= \sum_{x,y\in G}D_{\alpha}(x,y)\pi_x\pi_y^* &= \sum_{x,y\in G}\frac{\omega(x,y^{-1})}{\omega(y,y^{-1})}D_{\alpha}(x,y)\pi_{xy^{-1}}\\
        &= \sum_{x,z\in G}\frac{\omega(x,x^{-1}z)}{\omega(z^{-1}x,x^{-1}z)}D_{\alpha}(x,z^{-1}x)\pi_z.
    \end{align*}
    Comparing the coefficients we get 
    \[
    \sum_{x\in G} \frac{\omega(x,x^{-1}z)}{\omega(z^{-1}x,x^{-1}z)}D_{\alpha}(x,z^{-1}x) = \delta_{1,z}
    \]
    for all $z\in G$.
\end{proof}

%%%%%%%%%%
\section*{Acknowledgements}
%%%%%%%%%%

P.C.\ and U.F.\ were supported by the ANR Project No.\ ANR-19-CE40-0002. B.V.R. Bhat is suppported by the J C Bose Fellowship JBR/2021/000024 of SERB(India)

%%%%%%%%%%%%%%%%%%%%%%

\end{document}